\documentclass[conference,letterpaper]{IEEEtran}

\usepackage[utf8]{inputenc}
\usepackage[T1]{fontenc}
\usepackage{url}
\usepackage{ifthen}
\usepackage{cite}
\usepackage[cmex10]{amsmath} 
\usepackage[mathscr]{eucal}
\usepackage{amsmath,amsthm,amssymb,bm,bbm,dsfont,xfrac,braket,cases,comment}
\usepackage{flushend}
\usepackage{thmtools}
\usepackage{thm-restate}
\usepackage{mathtools}\mathtoolsset{centercolon}
\mathtoolsset{showonlyrefs}
\usepackage[usenames,dvipsnames]{color}
\usepackage[draft]{hyperref}
\hypersetup{
	unicode=false,          
	pdftoolbar=true,        
	pdfmenubar=true,        
	pdffitwindow=false,     
	pdfstartview={FitH},    
	pdftitle={My title},    
	pdfauthor={Author},     
	pdfsubject={Subject},   
	pdfcreator={Creator},   
	pdfproducer={Producer}, 
	pdfkeywords={keyword1} {key2} {key3}, 
	pdfnewwindow=true,      
	colorlinks=true,        
	linkcolor=Red,          
	citecolor=ForestGreen,  
	filecolor=Magenta,      
	urlcolor=BlueViolet,    
}
\usepackage{doi}
\usepackage{url}
\usepackage{caption, subcaption}
\usepackage{enumitem}
\usepackage{cite}

\DeclareMathOperator{\Tr}{Tr}

\DeclareMathOperator{\e}{\mathrm{e}}

\newcommand{\be}{{\mathbf e}}

\newcommand{\tr}{\operatorname{Tr}}
\newcommand{\al}{{\alpha}}

\def\0{{\mathbf{0}}}
\def\1{{\mathbf{1}}}
\def\2{{\mathbf{2}}}
\def\3{{\mathbf{3}}}
\def\4{{\mathbf{4}}}
\def\5{{\mathbf{5}}}
\def\6{{\mathbf{6}}}

\def\7{{\mathbf{7}}}
\def\8{{\mathbf{8}}}
\def\9{{\mathbf{9}}}

\def\be{\begin{equation}}
	\def\ee{\end{equation}}
\def\bea{\begin{eqnarray}}
	\def\eea{\end{eqnarray}}

\theoremstyle{plain}
\newtheorem{theo}{Theorem} 
\newtheorem{prop}[theo]{Proposition} 
\newtheorem{lemm}[theo]{Lemma} 

\theoremstyle{definition}
\newtheorem{defn}[theo]{Definition} 

\theoremstyle{remark}
\newtheorem{remark}{Remark}[section]

\begin{document}
\title{{Privacy Amplification Against Quantum Side Information Via Regular Random Binning}}


\author{%
  \IEEEauthorblockN{Yu-Chen Shen$^{1}$, Li Gao$^{2}$, and Hao-Chung Cheng$^{1,3-6}$}
  \IEEEauthorblockA{$^1$Department of Electrical Engineering
    National Taiwan University, Taipei 10617, Taiwan (R.O.C.)\\
    $^{2}$Department of Mathematics,  University of Houston, Houston, TX 77204, USA\\
    $^{3}$Department of Mathematics,
  National Taiwan University, Taipei 10617, Taiwan (R.O.C.)\\
  $^{4}$Center for Quantum Science and Engineering,  National Taiwan University, Taipei 10617, Taiwan (R.O.C.)\\
  $^{5}$Physics Division, National Center for Theoretical Sciences, Taipei 10617, Taiwan (R.O.C.)\\
  $^{6}$Hon Hai (Foxconn) Quantum Computing Research Center, New Taipei City 236, Taiwan (R.O.C.)
	}
}


\maketitle

\begin{abstract}

We consider privacy amplification against quantum side information by using regular random binning as an effective extractor. 
For constant-type sources, we obtain error exponent and strong converse bounds in terms of the so-called quantum Augustin information. Via type decomposition, we then recover the error exponent for independent and identically distributed sources proved by Dupuis [arXiv:2105.05342].
As an application, we obtain an achievable secrecy exponent for classical-quantum wiretap channel coding in terms of the Augustin information, which solves an open problem in
[IEEE Trans.~Inf.~Theory, 65(12):7985, 2019].
Our approach is to establish an operational equivalence between privacy amplification and quantum soft covering; this may be of independent interest.
\end{abstract}

\section{Introduction}
\emph{Privacy amplification} (also called \emph{randomness extraction} in \cite{SIAM17})
is a vital protocol in classical and quantum cryptography for extracting randomness from a source partially leaked to environment \cite{ILL89, Ren05, WH13, Hay13, PR14, HW16, Hay16, Tom16, Dup21, HT15, Hay12_}.
Privacy amplification has been widely studied for its applications in security bounds (known as the leftover hash lemma) \cite{Ren05, TSS+11, Hay13, Tom16, Dup21}, random number generation \cite{Hay13},  wiretap channel coding \cite{Hay13,Hay11,Hay2112}, and others\cite{IEEE64,IEEE57,Proc465,Proc467,Tsurumaru2021EquivalenceOT}.
Much progress has been made on characterizations of the information leakage in privacy amplification. The achievability part has been established under various security criteria such as trace distance \cite{Ren05,MH333,Tom16,Dup21}, purified distance \cite{TH13,LYH23}, and the quantum relative entropy \cite{LYH23,YW66}.
Converse studies were shown under purified distance \cite{TH13,SD22,LYH23} and trace distance \cite{SGC22a}. 

For privacy amplification against quantum adversaries, suppose Alice and Bob (adversary) share a classical-quantum (c-q) state $\rho_{XB} := \sum_{x\in\mathcal{X}} p_X(x)\ket{x}\bra{x}\otimes \rho^x_B$, where Alice holds the classical system $X$ and Bob holds the quantum system $B$.
To each source $x\in\mathcal{X}$ held at Alice, Bob accesses to a density operator $\rho_B^x$ (i.e.~positive semi-definite operator with unit trace) as the quantum side information.
The goal of privacy amplification is to extract uniform randomness, say on system $Z$, such that it is independent of Bob's system
The common method of privacy amplification is for Alice to apply a random hash function on her system. In this paper, we consider the \emph{regular random binning function} \cite{Yas14} as the hash function defined as follows:
\begin{defn}[Regular random binning function]
A regular random binning function $h: \mathcal{X}\to \mathcal{Z}$ is a randomly chosen $k$-to-$1$ function, where $k = \frac{|\mathcal{X}|}{|\mathcal{Z}|}$ in this case. 
\end{defn}
We denote the implementation of the regular random binning $h$ on system $X$ by the following linear operation $\mathcal{R}^h_{X\to Z}$:
\begin{align}
    \mathcal{R}^h(\rho_{XB}) &:= \sum_{x\in\mathcal{X}} p_X(x)|h(x)\rangle \langle h(x)|\otimes \rho_B^x\\
    &= \sum_{z\in\mathcal{Z}} |z\rangle\langle z|\otimes \sum_{x\in h^{-1}(z)} p_X(x)\rho_B^x.
\end{align}
We adopt the \emph{trace distance} between resulted state and the ideal state as the security criterion:
\begin{align}
\frac{1}{2}\mathds{E}_h\left\| \mathcal{R}^h(\rho_{XB}) - \frac{\mathds{1}_Z}{|\mathcal{Z}|} \otimes \rho_B \right\|_1.
\end{align}

In this paper, we consider privacy amplification for the \emph{constant-type sources}.
Let $n\in\mathds{N}$ be a positive integer and let probability distribution $p$ to be an $n$-type satisfying $np(x) \in 0\cup\mathds{N}$ for every $x\in\mathcal{X}$.
The constant-type source contains equiprobable sequences of the same empirical distribution on $\mathcal{X}^n$, i.e.
\begin{align}
\Breve{p}_{X^n}(x^n) = \begin{cases} \frac{1}{|T^n_p|}, & \mbox{if }x^n \in T^n_p\\ 0. & \mbox{else } \end{cases}.\label{typeprob}
\end{align}
Here the \emph{type class} $T^n_p$ is defined later in \eqref{eq:type}.
The joint c-q state of type $p$, denoted as $\Breve{\rho}_{X^nB^n}$, is then
\begin{align} \label{eq:const-type}
\Breve{\rho}_{X^nB^n}= \sum_{x^n\in T^n_p} \frac{1}{|T^n_p|}\ket{x^n}\bra{x^n}\otimes \rho^{x^n}_{B^n},
\end{align}
where $\rho^{x^n}_{B^n}:= \otimes_{i=1}^n \rho_B^{x_i}$.
To extract randomness from $\Breve{\rho}_{X^nB^n}$, Alice applies the random binning function $\mathcal{R}^{h^n}: T^n_p \to \mathcal{Z}^n$ on her system $X^n$. The secrecy measure is then
\begin{align}
d_{\text{PA}}(\Breve{\rho}_{X^nB^n}, R_{\text{PA}}) := \frac{1}{2}\mathds{E}_{h^n}\left\|\mathcal{R}^{h^n}(\Breve{\rho}_{X^nB^n})- \frac{\mathds{1}_{Z^n}}{|\mathcal{Z}^n|}\otimes \Breve{\rho}_{B^n}\right\|_1,
\end{align}
where $R_{\text{PA}} = \frac1n\log{|\mathcal{Z}^n|}$ is the rate of the extracted randomness and $\Breve{\rho}_{B^n}$ is the marginal state of $\Breve{\rho}_{X^nB^n}$ on system $B^n$.

On the other hand, \emph{soft covering} originates from Wyner \cite{Wyn75-1} for studying common information. It is a useful tool for proving security of wiretap channel coding \cite{MJ13}, or doing other secrecy analysis and coding problems, e.g.~\cite{TS93, SM14}. In the quantum setting, soft covering has been studied in the context of compression and channel simulation \cite{DW03, DW05, LD09}. 

In this paper, we consider quantum soft covering via \emph{constant composition random codebook}.
Fix an product c-q channel $x^n \mapsto \rho_{B^n}^{x^n}$.
Given a constant-type source  $\breve{p}_{X^n}$ defined in \eqref{typeprob} as the prior, 
the induced  state at channel output is
\begin{align}
\breve{\rho}_{B^n} = \sum_{x^n\in T^n_p}\frac{1}{|T^n_p|}\rho^{x^n}_{B^n},
\end{align}
which is the marginal state of  $\Breve{\rho}_{X^nB^n}$ in \eqref{eq:const-type}.
The task of \emph{quantum soft covering} is to approximate $\breve{\rho}_{B^n}$ via certain random codebook.
In this paper, we consider a random codebook $\Breve{C}^n$ in which \emph{distinct} codewords are chosen uniformly from the type class $T^n_p$. 
We termed this the (constant-type) \emph{quantum soft covering without repetition}.

\begin{defn}[Constant composition random codebook without repetition] \label{defn:sc}
For any type class $T^n_p$ and integer $M\in\mathds{N}$, the random codebook $\Breve{\mathcal{C}}^n$ without repetition with cardinality $M$ obeys:
\begin{align} \label{eq:prop_constant_type}
&\Pr\left[ \Breve{\mathcal{C}}^n = \left\{ x^n(1),\ldots,x^n(M) \right\} \right] \\
&=
\begin{cases}
\tfrac{ (|\mathcal{X}| - M )! }{ |\mathcal{X}| {!}   } & x^n(i), x^n(j) \in T^n_p, x^n(i)\neq x^n({j}), \forall\, i\neq j \\
0 & \text{otherwise}
\end{cases}.
\end{align}
\end{defn}

Via the codebook ${\Breve{\mathcal{C}}^n}$ 
in Definition~\ref{defn:sc}, the codebook-induced state is $\rho^{\Breve{\mathcal{C}}^n}_{B^n} = \frac{1}{|\Breve{\mathcal{C}}^n|}\sum_{x^n\in\Breve{\mathcal{C}}^n}\rho^{x^n}_{B^n}$.
We use the trace distance between the approximated state $\rho^{\Breve{\mathcal{C}}^n}_{B^n}$ and the true marginal state $\breve{\rho}_{B^n}$ to quantify the error of approximation:
\begin{align}
d_{\text{SC}}(\Breve{\rho}_{X^nB^n}, R_{\text{SC}}) := \frac{1}{2} \mathds{E}_{\Breve{\mathcal{C}}^n}\left\|\rho^{\Breve{\mathcal{C}}^n}_{B^n} - \Breve{\rho}_{B^n}\right\|_1,
\end{align}
where $R_{\text{SC}} := \frac1n \log|\Breve{\mathcal{C}}^n|$,
and the expectation is with respect to \eqref{eq:prop_constant_type}.

The main goal of this paper is to establish achievability and strong converse for privacy amplification on constant-type sources. 
Note that the achievability has been studied  by 
Mojahedian \textit{et al.} \cite{Bei18} assuming that the side information is classical (i.e.~all $\{\rho_B^x\}_{x\in\mathcal{X}}$ mutually commute).
We then aim at the scenario of quantum side information to address the open problems in \cite{Bei18}.
To that end, we first prove an \emph{operational equivalence} between privacy amplification using regular random binning and quantum soft covering without repetition (Theorem~\ref{equiv}):
\begin{align}\label{eq:equ}
d_{\text{PA}}(\Breve{\rho}_{X^nB^n}, R) = d_{\text{SC}}\left(\Breve{\rho}_{X^nB^n}, \sfrac{\log{|T^n_p|}}{n}- R\right).
\end{align}
Next, we establish an achievability bound and a strong converse bound for quantum soft covering without repetition (Theorems~\ref{scdirect} and \ref{scconverse}). 
Hence, by applying those bounds and using the equivalence in \eqref{eq:equ}, 
we obtain the following results for privacy amplification on constant-type sources (Theorems~\ref{PAdirect} and \ref{PAconverse}):
\begin{align}
\begin{cases} 
d_{\text{PA}}(\Breve{\rho}_{X^nB^n},R)& \dot{\leq} \e^{-n \frac{\alpha-1}{\alpha}\left(H(p) - \Breve{I}^*_{\alpha}(X:B)_{\rho}- R \right)}\\
1 - d_{\text{PA}}(\Breve{\rho}_{X^n B^n},R) &\dot{\leq}\e^{-n\frac{1-\alpha'}{\alpha'}\left( R - H(p) + \Breve{I}^{\uparrow}_{2-\sfrac{1}{\alpha'}}(X:B)_\rho \right)}
\end{cases}
\end{align}
for any $\alpha \in (1,2)$ and $\alpha' \in (\sfrac{1}{2},1)$. 
Here, ``$\dot{\leq} $" means that the inequality is up to a polynomial prefactor in $n$ and $H(p):=-\sum_x p(x)\log p(x)$ is the Shannon entropy.
The established exponents are expressed by the \emph{Augustin-type information} \cite{DW14, MO18, CGH18} defined later in \eqref{eq:Sand_Augustin_info} and \eqref{eq:Petz_Augustin_info}.
The achievability (Theorem~\ref{PAdirect}) then generalizes \cite[Theorem 25]{Bei18} by Mojahedian \textit{et al.}~to the scenario of quantum side information.

Our results have several implications.
Firstly, the above bounds show that the \emph{quantum conditional entropy} (later defined in \eqref{eq:conditional_entropy}) is the 
fundamental limit of the maximal  extractable randomness as well as the exponential strong converse rate for privacy amplification on constant-type sources.
Secondly, we show that the established exponents expressed in terms of the Augustin-type information are larger than that of the previous works \cite{Dup21, SGC22b} (see Proposition~\ref{prop:bound}). 
This then demonstrates sharper concentration of $d_{\text{PA}}$ (resp.~$1-d_{\text{PA}}$) to $0$ (resp.~$1$) for constant-type sources.

Moreover, using our achievability bound for constant-type sources together with the type decomposition of independent and identically distributed (i.i.d.)~c-q state, we obtain an achievability bound of privacy amplification on i.i.d~sources (Proposition~\ref{PAiid}):
\begin{align}
& \min_{h^n: \mathcal{X}^n \to \mathcal{Z}^n }\frac{1}{2}\left\|\mathcal{R}^{h^n}(\rho^{\otimes n}_{XB})-\frac{1}{|\mathcal{Z}^n|}\sum_{z^n\in \mathcal{Z}^n}\ket{z^n}\bra{z^n}\otimes \rho_B^{\otimes n}\right\|_1\\
&{\dot{\leq}} \e^{-n \min\limits_{q}\sup\limits_{\alpha\in(1,2)}D(q\Vert p)+\frac{\alpha-1}{\alpha}\big( H(q) - \Breve{I}^*_{\alpha}(X:B)_{\rho^{q}}- R \big)},
\end{align}
where the minimization is over all types, $\rho^q_{XB} = \sum_{x\in \mathcal{X}}q(x)\ket{x}\bra{x}\otimes \rho^x_B$, and $D(\cdot \Vert \cdot)$ is the Kullback--Leibler divergence.
Via an entropic relation proved in \cite[Theorem~10]{CHDH2-2018}, the above error exponent matches the one obtained by Dupuis \cite[Theorem~8]{Dup21} for i.i.d.~sources. Hence, we provide an alternative approach to privacy amplification on i.i.d.~sources
using regular random binning and constant-type sources as a bridge.

Lastly, our achievability bound for constant-type sources applies to private classical communication over classical-quantum wiretap channels $x\mapsto \sigma_{BE}^x$.
We exploit constant composition codes with regular random binning as the coding strategy.
While reliably sending classical information of rate $R$ to Bob (with asymptotically vanishing errors),
we show that the secrecy exponent of information leakage to the quantum eavesdropper (holding system $E$) is at least (Theorem~\ref{theo:wireach}):
\begin{align}
    \sup_{\alpha\in(1,2)}\frac{\alpha-1}{\alpha}\left(I(X:B)_{\sigma} - \Breve{I}^*_{\alpha}(X:E)_{\sigma}- R \right).
\end{align}
where $\sigma_{XBE} = \sum_{x\in \mathcal{X}}p_X(x)\ket{x}\bra{x}\otimes \sigma^x_{BE}$, and 
$I(X:B)_{\sigma}$ is the quantum mutual information,
Further, the secrecy exponent is positive if and only if $R < {I(X\!:\!B)_{\sigma} \!-\! {I}(X\!:\!E)_{\sigma}}$. 
Our result also answers the open problem addressed in \cite[Theorem~28]{Bei18}; namely, random binning is effective for secrecy against quantum side information.

The paper is structured as follows. Section~\ref{sec:notation} reviews the necessary background on information-theoretic quantities. In section \ref{SC}, we prove the achievability and strong converse bound of quantum soft covering without repetition. In section \ref{sec:SCtoPA}, we show the operational equivalence between privacy amplification via regular random binning and quantum soft covering without repetition. In section \ref{PAtype}, we establish achievability and strong converse bounds for privacy amplification. Section \ref{wiretap} shows the application on bounding information leakage in c-q wiretap channel coding. 
We conclude the paper in section \ref{conclude}.

\section{Notations and information quantities}\label{sec:notation}
We denote by $\mathcal{S}(\mathcal{H})$ the set of density operators on Hilbert space $\mathcal{H}$.
Let $\mathscr{P}(\mathcal{X})$ be the set of all probability distributions on a finite alphabet $\mathcal{X}$.
For $\mathsf{p}\geq 1$, the Schatten $\mathsf{p}$-norm is $
\left\|M\right\|_{\mathsf{p}} := (\Tr[|M|^\mathsf{p}])^{\sfrac{1}{\mathsf{p}}}$.
For any $M\in\mathds{N}$, we let $[M]:=\{1,\ldots, M\}$.

For $\alpha\in(0,\infty)\backslash 1$, the order-$\alpha$ Petz--R\'enyi divergence $D_\alpha$ \cite{Pet86} and the sandwiched R\'enyi divergence $D^*_{\alpha}$ \cite{MDS+13, WWY14} are defined as
\begin{align}
D_\alpha(\rho\|\sigma)&:=\frac{1}{\al-1}\log\tr\left[\rho^{\alpha}\sigma^{1-\alpha}\right],\\
D^*_\alpha(\rho\|\sigma)&:= \frac{1}{\alpha-1}\log\left\|\sigma^{\frac{1-\alpha}{2\alpha}}\rho\sigma^{\frac{1-\alpha}{2\alpha}}\right\|_{\alpha}^{\alpha}
\end{align}
Note both divergences converge to the Umegaki relative entropy \cite{Ume62} $D(\rho\|\sigma) := \Tr\left[ \rho (\log \rho - \log \sigma )\right]$ when $\alpha\to 1$ (see e.g.~\cite[Lemma 3.5]{MO14}).

For a classical-quantum (c-q) state $\rho_{XB} = \sum_{x\in\mathcal{X}}p_X(x)\ket{x}\bra{x}\otimes \rho^x_B$, we define the following  \emph{order-$\alpha$ sandwiched Augustin information}:
\begin{align} \label{eq:Sand_Augustin_info}
\begin{split} 
&\Breve{I}^*_{\alpha}(X:B)_{\rho}  \\
&:= \inf_{\sigma_B\in \mathcal{S}(\mathcal{H}_B)}\frac{\alpha}{\alpha-1}\sum_{x\in \mathcal{X}}p_X(x)\log\left\|\sigma_B^{\frac{1-\alpha}{2\alpha}}\rho^x_B\sigma_B^{\frac{1-\alpha}{2\alpha}}\right\|_{\alpha},
\end{split}
\end{align}
as well as the (Augustin-like) Petz-type information quantity:
\begin{align} \label{eq:Petz_Augustin_info}
\Breve{I}^{\uparrow}_{\alpha}(X:B)_{\rho} &:= \sum_{x\in \mathcal{X}}p_X(x)D_{\alpha}(\rho^x_{B}\|\rho_B).
\end{align}
The above two information quantities converge to the \emph{quantum mutual information} $I(X:B)_{\rho}:= D(\rho_{XB}\|\rho_X\otimes\rho_B)$ as $\alpha \to 1$ (see e.g.~\cite{MO17}).

For any $n\in\mathds{N}$, we call a probability distribution $p$ on $\mathcal{X}$ an \emph{$n$-type} if $n p(x) \in \{0\}\cup \mathds{N}$ for all $x\in \mathcal{X}$.
For an $n$-type $p$, we define the \emph{type class} of $p$ as the set of sequences in $\mathcal{X}^n$ with its empirical distributions being $p$, i.e.
\begin{align}\label{eq:type}
T^n_p := \left\{x^n\in \mathcal{X}^n: \frac{1}{n}\sum\nolimits_{ i\in[n] }\mathbf{1}_{\{x=x_i\}}= p(x)\right\}.
\end{align}
The size of $T^n_p$ is bounded as follows \cite{CK11}:
\begin{align}
(n+1)^{-|\mathcal{X}|}\e^{nH(p)} \leq |T^n_p| \leq \e^{nH(p)} \label{typebd}.
\end{align}
The set of all $n$-types on $\mathcal{X}$, denoted as $\mathscr{P}_n(\mathcal{X})$, is bounded by \cite{CK11}
\begin{align}
|\mathscr{P}_n(\mathcal{X})| \leq (n+1)^{|\mathcal{X}|}. \label{typenum}
\end{align}

\section{Quantum Soft Covering without Repetition}\label{SC}
In \cite{CG22}, part of our authors proved the direct bound and strong converse bound for quantum soft covering using a pairwise independent constant composition random  codebook.
Following similar reasoning, we prove an achievable error exponent and a strong converse bound for quantum soft covering using the new random codebook without repetition in Definition~\ref{defn:sc}.

\subsection{Achievability bound}
\begin{theo}\label{scdirect}
For any $n \in \mathds{N}$ and $\rho_{XB} = \sum_{x\in \mathcal{X}}p_X(x)\ket{x}\bra{x}\otimes \rho^x_B$, where $p_X$ is an $n$-type. 
Using the random codebook in Definition~\ref{defn:sc}, the expected trace distance between the output state $\rho^{\Breve{\mathcal{C}}^n}_{B^n}$ and the true marginal state $\Breve{\rho}_{B^n}$ is upper bounded by:
\begin{align}
d_{\textnormal{SC}}(\Breve{\rho}_{X^nB^n}, R) \leq \e^{-n \sup_{\alpha\in(1,2)}\frac{1-\alpha}{\alpha}\left(\Breve{I}^*_{\alpha}(X:B)_{\rho}-R\right)},
\end{align}
where $R:= \frac{1}{n}\log|\Breve{\mathcal{C}}^n|$ is the rate of the codebook size.
\end{theo}
\noindent The proof is deferred to Appendix \ref{sec:proof_scdirect}.

The main idea is adapted from \cite{CG22} except that we show the cross-moment between random codewords using a codebook given in Definition~\ref{defn:sc} is negative instead of zero compared with the pairwise independent codebook used in \cite{CG22}.

\subsection{Strong converse bound}
\begin{theo}\label{scconverse}
For any $n \in \mathds{N}$ and $\rho_{XB} = \sum_{x\in \mathcal{X}}p_X(x)\ket{x}\bra{x}\otimes \rho^x_B$, where $p_X$ is an $n$-type. 
Using the random codebook in Definition~\ref{defn:sc}, the expected trace distance between the output state $\rho^{\Breve{\mathcal{C}}^n}_{B^n}$ and the true marginal state $\Breve{\rho}_{B^n}$ is lower bounded by:
\begin{align}
&d_{\textnormal{SC}}(\Breve{\rho}_{X^nB^n}, R) \\
&\geq 1- 4(n+1)^{|\mathcal{X}|}\e^{-n\sup\limits_{\alpha\in (\sfrac12, 1)}\frac{1-\alpha}{\alpha}\left(\Breve{I}^{\uparrow}_{2-\sfrac{1}{\alpha}}(X:B)_\rho -R\right)},
\end{align}
where $R= \frac{1}{n}\log|\Breve{\mathcal{C}}^n|$ is the rate of the codebook size.
\end{theo}
\noindent The proof is deferred to Appendix \ref{sec:proof_scconverse}. 

\section{Operational Equivalence between Privacy Amplification and Soft Covering} \label{sec:SCtoPA}
\begin{theo}[Equivalence between soft covering and privacy amplification]\label{equiv}
For any $n\in\mathds{N}$, $n$-type $p_X$, and 
$\Breve{\rho}_{X^nB^n} = \sum_{x^n\in T^n_p} \frac{1}{|T^n_p|}\ket{x^n}\bra{x^n}\otimes \rho^{x^n}_{B^n}$, 
the following holds for all $R$:
\begin{align}
d_{\textnormal{PA}}(\Breve{\rho}_{X^nB^n}, R) = d_{\textnormal{SC}}\left(\Breve{\rho}_{X^nB^n}, \sfrac{\log{|T^n_p|}}{n}- R\right).  
\end{align}
\end{theo}
\begin{proof}
Let $\log |\mathcal{Z}^n| = nR$. With the definition of the regular random binning function $\mathcal{R}^{h^n}$,
\begin{align}
&\mathds{E}_{h^n}\left\|\mathcal{R}^{h^n}(\Breve{\rho}_{X^nB^n})- \frac{\mathds{1}_{Z^n}}{|\mathcal{Z}^n|}\otimes \Breve{\rho}_{B^n}\right\|_1\\
& = \mathds{E}_{h^n} \bigg\|\frac{1}{|T^n_p|}\sum_{z^n\in \mathcal{Z}^n}\ket{z^n}\bra{z^n}\otimes \left(\sum_{x^n: h^n(x^n) = z^n}\rho^{x^n}_{B^n}\right)\\
& \quad- \frac{1}{|\mathcal{Z}^n|}\sum_{z^n\in \mathcal{Z}^n}\ket{z^n}\bra{z^n}\otimes \Breve{\rho}_{B^n}\bigg\|_1\\
& = \mathds{E}_{h^n} \frac{1}{|\mathcal{Z}^n|}\sum_{z^n\in \mathcal{Z}^n}\bigg\|\frac{1}{\sfrac{|T^n_p|}{|\mathcal{Z}^n|}}\sum_{x^n: h^n(x^n) = z^n}\rho^{x^n}_{B^n}- \Breve{\rho}_{B^n}\bigg\|_1\\
& =  \frac{1}{|\mathcal{Z}^n|}\sum_{z^n\in \mathcal{Z}^n}\mathds{E}_{h^n}\bigg\|\frac{1}{\sfrac{|T^n_p|}{|\mathcal{Z}^n|}}\sum_{x^n: h^n(x^n) = z^n}\rho^{x^n}_{B^n}- \Breve{\rho}_{B^n}\bigg\|_1\\
& = \mathds{E}_{h^n}\left\|\frac{1}{\sfrac{|T^n_p|}{|\mathcal{Z}^n|}}\sum_{x^n: h^n(x^n) = z_1^n}\rho^{x^n}_{B^n}- \Breve{\rho}_{B^n}\right\|_1\\
& = \mathds{E}_{\Breve{\mathcal{C}}^n}\left\|\rho^{\Breve{\mathcal{C}}^n}_{B^n}- \Breve{\rho}_{B^n}\right\|_1,
\end{align}
where $z^n_1$ is an arbitrary element in $\mathcal{Z}^n$ and $|\Breve{\mathcal{C}}^n| = \frac{|T^n_p|}{|\mathcal{Z}^n|}$ is the size of the  constant composition codebook without repetition (see Definition~\ref{defn:sc}). With the above equality, the secrecy of privacy amplification with rate $R$ is the same as the approximation error of quantum soft covering without repetition with the rate being $\frac1n\log{\frac{|T^n_p|}{|\mathcal{Z}^n|}} = \frac1n \log{|T^n_p|}-R$, which matches our statement in the theorem.
\end{proof}

\section{bounding privacy amplification \\ from soft covering}\label{PAtype}

Throughout this section, we fix a classical-quantum state $\rho_{XB} = \sum_{x\in\mathcal{X}} p(x)|x\rangle \langle x| \otimes \rho_B^x$.
In Section~\ref{sec:PA_constant-type}, we establish both achievability and strong converse of privacy amplification on constant-type sources.
In Section~\ref{sec:i.i.d.}, we recover the achievability of privacy amplification on i.i.d.~sources $\rho_{XB}^{\otimes n}$.

\subsection{Privacy amplification on constant-type sources} \label{sec:PA_constant-type}

By choosing $|\Breve{\mathcal{C}}^n| = \sfrac{|T^n_p|}{|\mathcal{Z}^n|}$,
Theorem~\ref{equiv} shows that the following relation holds between privacy amplification for constant-type sources $\Breve{\rho}_{X^n B^n}$ introduced in \eqref{eq:const-type} and soft covering without repetition:
\begin{align}
\mathds{E}_{h^n}\bigg\|\mathcal{R}^{h^n}(\Breve{\rho}_{X^nB^n}) \!-\! \frac{\mathds{1}_{Z^n}}{|\mathcal{Z}^n|}\otimes \Breve{\rho}_{B^n}\bigg\|_1 
& =  \mathds{E}_{\Breve{\mathcal{C}}^n}\left\|\rho^{\Breve{\mathcal{C}}^n}_{B^n} \!-\! \Breve{\rho}_{B^n}\right\|_1.
\end{align}
With this relation, we establish an achievability bound for privacy amplification on constant-type sources $\Breve{\rho}_{X^nB^n} $
using Theorem~\ref{scdirect} and the bounds of $|T^n_p|$ shown in \eqref{typebd}.

\begin{theo}[Achievability for privacy amplification on constant-type sources]\label{PAdirect}
Consider any $n\in\mathds{N}$, any $n$-type $p$, and a constant-type source $
\Breve{\rho}_{X^nB^n} = \sum_{x^n\in T^n_p} \frac{1}{|T^n_p|}\ket{x^n}\bra{x^n}\otimes \rho^{x^n}_{B^n}$ defined in \eqref{eq:const-type}.
Then,
\begin{align}
&d_{\textnormal{PA}}(\Breve{\rho}_{X^nB^n}, R) \\
&\leq \e^{-n \sup_{\alpha\in(1,2)}\frac{\alpha-1}{\alpha}\left(\frac1n\log{|T^n_p|}-\Breve{I}^*_{\alpha}(X:B)_{\rho}- R\right)}\\ 
& \leq (n+1)^{\frac{1}{2}|\mathcal{X}|}\e^{-n \sup_{\alpha\in(1,2)}\frac{\alpha-1}{\alpha}\left(H(p) - \Breve{I}^*_{\alpha}(X:B)_{\rho}- R\right)}.
\end{align}
\end{theo}
\begin{remark}
Theorem~\ref{PAdirect} generalizes the previous result \cite[Theorem 25]{Bei18} considering classical side information to the scenario of quantum side information.
\end{remark}

We remark that Dupuis established a very nice achievability bound for privacy amplification against quantum side information \cite{Dup21}. In the i.i.d.~setting where the underlying source is an $n$-fold product of $\rho_{XB}$, \cite[Theorem~8]{Dup21} shows that

\begin{align} \label{eq:Dup21}
d_{\textnormal{PA}}({\rho}_{X B}^{\otimes n}, R) 
\leq \e^{ - n \sup_{\alpha\in(1,2)}  \frac{\alpha-1}{\alpha}\left(H^*_{\alpha}(X|B)_{\rho}- R\right) },
\end{align}
where $H^*_{\alpha}(X{\,|\,}B)_{\rho} := - \inf_{\sigma_B\in \mathcal{S}(\mathcal{H}_B)}D^*_{\alpha}(\rho_{XB}\|\mathds{1}_X\otimes \sigma_B)$ is the \emph{sandwiched conditional R\'enyi entropy}.
It is thus interesting to compare our Theorem~\ref{PAdirect} with \eqref{eq:Dup21} (under the same prior $p$).
In Proposition~\ref{prop:bound} below, we show that, indeed, the (achievable) error exponent for constant-type sources $\Breve{\rho}_{X^nB^n}$ is larger than that of the i.i.d.~source ${\rho}_{X B}^{\otimes n}$.
This means that under the same rate $R$ and prior $p$, one achieves less information leakage for the constant-type sources compared to the i.i.d.~sources (albeit with an additional polynomial prefactor).
Later in Section~\ref{sec:i.i.d.}, we will show how to recover the scenario of i.i.d.~sources, i.e.~the error exponent in \eqref{eq:Dup21}, via Theorem~\ref{PAdirect}.

\begin{prop} \label{prop:bound}
    For any classical-quantum state $\rho_{XB} = \sum_{x\in\mathcal{X}} p(x)|x\rangle \langle x| \otimes \rho_B^x$ and $\alpha>1$,
    \begin{align}
    H(p) - \Breve{I}^*_{\alpha}(X:B)_{\rho} \geq H^*_{\alpha}(X{\,|\,}B)_{\rho}.
    \end{align}
\end{prop}
\noindent The proof is deferred to Appendix~\ref{app:bound}.


\medskip
Similar to the achievability, we bound a strong converse bound for privacy amplification using Theorem~\ref{scconverse}.
\begin{theo}[Strong converse for privacy amplification on constant-type sources]\label{PAconverse}
For any $n\in\mathds{N}$, any $n$-type $p$, and a constant-type source
$\Breve{\rho}_{X^nB^n} = \sum_{x^n\in T^n_p} \frac{1}{|T^n_p|}\ket{x^n}\bra{x^n}\otimes \rho^{x^n}_{B^n}$ defined in \eqref{eq:const-type}, the following strong converse bound holds:
\begin{align}
&d_{\textnormal{PA}}(\Breve{\rho}_{X^nB^n}, R)\\
& \geq 1- 4(n+1)^{|\mathcal{X}|}\e^{-n\sup\limits_{\alpha\in (\sfrac12, 1)}\frac{1-\alpha}{\alpha}\left(\Breve{I}^{\uparrow}_{2-\sfrac{1}{\alpha}}(X:B)_\rho -\frac{\log|T^n_p|}{n}+R\right)}\\
&\geq 1- 4(n+1)^{|\mathcal{X}|}\e^{-n\sup\limits_{\alpha\in (\sfrac12, 1)}\frac{1-\alpha}{\alpha}\left(\Breve{I}^{\uparrow}_{2-\sfrac{1}{\alpha}}(X:B)_\rho -H(p)+R\right)}.
\end{align} 
\end{theo}
\begin{remark}
Following the same reasoning in Proposition~\ref{prop:bound},
the derived exponent in terms of the Augustin information for the strong converse bound is larger than the exponent of our previous result \cite[Theorem~1]{SGC22a}, wherein the exponent is expressed by the Petz-type conditional R\'enyi entropy:
\begin{align}
\sup_{\alpha\in (\sfrac12, 1)}\frac{1-\alpha}{\alpha} \left( R - H_{2-\sfrac{1}{\alpha}}^\downarrow (X|B)_\rho  \right).
\end{align}
Here $H_\alpha^\downarrow (X|B)_\rho := - D_\alpha (\rho_{XB} \Vert \mathds{1}\otimes \rho_B)$ for  $\rho_{XB}= \sum_{x\in \mathcal{X}}p(x)\ket{x}\bra{x}\otimes \rho^x_B$.
However, there is an polynomial prefactor in Theorem~\ref{PAconverse} as expected.
\end{remark}

As mentioned in Section~\ref{sec:notation}, it holds that
\begin{align} \lim_{\alpha\to 1}
\Breve{I}_\alpha(X:B)_\rho = \Breve{I}_\alpha^\uparrow (X:B)_\rho = I(X:B)_\rho.
\end{align}
Theorems~\ref{PAdirect} and \ref{PAconverse} shows that the \emph{quantum conditional entropy}
\begin{align} \label{eq:conditional_entropy}
H(X|B)_\rho := H(p) - I(X:B)_\rho
\end{align}
serves as the fundamental limit of the maximal  extractable randomness as well as the exponential strong converse rate.

\subsection{Privacy amplification for i.i.d.~sources} \label{sec:i.i.d.}
Consider a $n$-fold i.i.d.~quantum state $\rho_{X^nB^n} = \rho^{\otimes n}_{XB} = \left(\sum_{x\in \mathcal{X}}p(x)\ket{x}\bra{x}\otimes \rho^x_B\right)^{\otimes n}$.
Denote $\rho^{(q)}_{X^nB^n}$ as a c-q state in which $x^n$ is distributed uniformly on $T^n_q$; namely $\rho^{(q)}_{X^nB^n}: = \frac{1}{|T^n_q|}\sum_{x^n\in T^n_q} \ket{x^n}\bra{x^n} \otimes \rho^{x^n}_{B^n}$.
The type decomposition of $\rho_{XB}^{\otimes n}$ is shown as follows:
\begin{align}
\rho^{\otimes n}_{XB} = \sum_{q\in \mathscr{P}_n(\mathcal{X})}\Pr[x^n\in T^n_q]\rho^{(q)}_{X^nB^n}. \label{rhoXBdecompose}
\end{align}
Similarly, the marginal state $\rho_B^{\otimes n}$ admits the type decomposition:
\begin{align}
\rho^{\otimes n}_{B} = \sum_{q\in \mathscr{P}_n(\mathcal{X})}\Pr[x^n\in T^n_q]\rho^{(q)}_{B^n}.\label{rhoBdecompose}
\end{align}
Using \eqref{rhoXBdecompose}, \eqref{rhoBdecompose}, as well as Theorem~\ref{PAdirect}, the minimum trace distance between the resulted state and the ideal state in privacy amplification on i.i.d.~sources can be upper bounded as follows.

\begin{prop}[Achievability for privacy amplification on i.i.d.~sources]\label{PAiid}
For any $n\in\mathds{N}$, an $n$-type $p$, and a classical-quantum state $\rho_{XB}= \sum_{x\in \mathcal{X}}p(x)\ket{x}\bra{x}\otimes \rho^x_B$, we have,
\begin{align}
& \min_{h^n}\frac{1}{2}\left\|\mathcal{R}^{h^n}(\rho^{\otimes n}_{XB})-\frac{1}{|\mathcal{Z}^n|}\sum_{z^n\in \mathcal{Z}^n}\ket{z^n}\bra{z^n}\otimes \rho_B^{\otimes n}\right\|_1\\
&\leq (n+1)^{\frac{3}{2}|\mathcal{X}|}\\
&\quad \cdot \e^{-n \min\limits_{q\in \mathscr{P}_n(\mathcal{X})}\sup\limits_{\alpha\in(1,2)}D(q \Vert p)+\frac{\alpha-1}{\alpha}\left( H(q) - \Breve{I}^*_{\alpha}(X:B)_{\rho^{(q)}}- R \right)},
\end{align}
where $\rho^{(q)}_{XB} = \sum_{x\in \mathcal{X}}q(x)\ket{x}\bra{x}\otimes \rho^x_B$.
\end{prop}

Below we explain the error exponent achieved in Proposition~\ref{PAiid}.
Via an entropic relation established in \cite[Theorem~10]{CHDH2-2018}:
\begin{align}
&\min\limits_{q\in \mathscr{P}(\mathcal{X})}\sup\limits_{\alpha\in(1,2)}D(q \Vert p)+\frac{\alpha-1}{\alpha}\left( H(q) \!-\! \Breve{I}^*_{\alpha}(X\!:\!B)_{\rho^{(q)}} \!-\! R \right)
\notag
\\
&= \frac{\alpha-1}{\alpha}\left( H_\alpha^*(X\mid B)_\rho - R \right).
\end{align}
In other words, the secrecy exponent for privacy amplification shown in Proposition~\ref{PAiid} matches that of the previous work \cite[Theorem~8]{Dup21}, which is \eqref{eq:Dup21}. 
This then provides an alternative proof to achieving the same secrecy exponent of privacy amplification on i.i.d.~sources (in the asymptotic limit). 
Note that the result in \cite{Dup21} is still stronger since there is a polynomial prefactor in Proposition~\ref{PAiid}.

\begin{proof}[Proof of Proposition~\ref{PAiid}]
The derivation of our claim is as follows,
\begin{align}
& \min_{h^n}\frac{1}{2}\left\|\mathcal{R}^{h^n}(\rho^{\otimes n}_{XB})-\frac{1}{|\mathcal{Z}^n|}\sum_{z^n\in \mathcal{Z}^n}\ket{z^n}\bra{z^n}\otimes \rho_B^{\otimes n}\right\|_1\\
& = \min_{h^n}\frac{1}{2}\left\|\sum_{q\in \mathscr{P}_n(\mathcal{X})}\Pr\left[x^n\in T^n_q\right]\right.\\
&\quad \quad \quad \cdot \left.\left(\mathcal{R}^{h^n}\left(\rho^{(q)}_{X^nB^n}\right) -\frac{1}{|\mathcal{Z}^n|}\sum_{z^n\in \mathcal{Z}^n}\ket{z^n}\bra{z^n}\otimes \rho^{(q)}_{B^n}\right)\right\|_1\\
& \overset{\textnormal{(a)}}{\leq} \sum_{q\in \mathscr{P}_n(\mathcal{X})}\Pr\left[x^n\in T^n_q\right] \\
&\quad \quad \cdot\frac{1}{2}\min_{h_q}
\bigg\|\mathcal{R}^{h_q}(\rho^{(q)}_{X^nB^n})-\frac{1}{|\mathcal{Z}^n|}\sum_{z^n\in \mathcal{Z}^n}\ket{z^n}\bra{z^n}\otimes \rho^{(q)}_{B^n}\bigg\|_1\\
& \overset{\textnormal{(b)}}{\leq} \sum_{q\in \mathscr{P}_n(\mathcal{X})}\Pr\left[x^n\in T^n_q\right]d_{\text{PA}}\left(\rho^{(q)}_{X^nB^n}, \log|\mathcal{Z}^n|\right)\\
& \leq \left|\mathcal{P}_n(\mathcal{X}^n)\right| \max_{q\in \mathscr{P}_n(\mathcal{X})}\Pr\left[x^n\in T^n_q\right]d_{\text{PA}}\left(\rho^{(q)}_{X^nB^n}, \log|\mathcal{Z}^n|\right)\\
& \overset{\textnormal{(c)}}{\leq}  \left|\mathcal{P}_n(\mathcal{X}^n)\right| \\
& \quad\cdot\max_{q\in \mathscr{P}_n(\mathcal{X})} \e^{-n \sup\limits_{\alpha\in(1,2)}D( q \Vert p)+\frac{\alpha-1}{\alpha}\left( \log\frac{|T^n_q|}{n} - \Breve{I}^*_{\alpha}(X:B)_{\rho^{(q)}}- R \right)},
\end{align}
where $h_q: T^n_q \to \mathcal{Z}$ is the regular random binning function associated to the type $q$.
For inequality (a), we use triangle inequality and that the set of all combined function $\{\mathcal{R}^{h_q} : q\in \mathscr{P}_n(\mathcal{X})\}$ is a subset of all set of 
 ${\mathcal{R}^{h^n}}$. 
 In (b), note that the expectation value is at least as large as minimum value.
 In (c), we applied Theorem~\ref{PAdirect}.
 The proof is concluded by invoking~\eqref{typebd} and \eqref{typenum}.
\end{proof}

\section{Application: the wiretap channel coding}\label{wiretap}
We now apply our results to classical-quantum (c-q) wiretap channels and obtain an estimate of the information leak to the quantum eavesdropper. 
Consider a c-q wiretap channel $\mathcal{N}_{X\to BE}: x \mapsto \sigma_{BE}^x$,
where Alice sends a classical symbol $x\in \mathcal{X}$, and the channel output states received by Bob and Eve are respectively the marginal states $\sigma_B^x \in \mathcal{S}(\mathcal{H}_B)$ and $\sigma_E^x \in \mathcal{S}(\mathcal{H}_E)$. The goal of Alice is to communicate equiprobable classical messages from a message set $[M]$ to Bob over $\mathcal{N}_{X\to BE}$, without leaking too much information to the eavesdropper, Eve.
%
%
In the following, we provide a protocol for this purpose. 
Without loss of generality, we consider any finite blocklength $n\in \mathds{N}$ and any $n$-type $p$.
\begin{enumerate}[label=\arabic*.]
    \item\label{label:3} 
    Let $f: T^n_p \to [M]\times [L]\times [K]$ be a random uniform one-to-one mapping, where $T^n_p$ is the type class of $p$. 
    Alice and Bob generate the function $f$ and claim it publicly. 
    \item Alice and Bob randomly choose a $k\in [K]$, claim it publicly, and fix it when using $f$ in next steps.
    \item For the message $m\in[M]$ Alice wanted to send, she uniformly chooses an element $x^n$ from $\left\{x^n:f(x^n) = (m,\cdot,k) \right\}$. 
    \item Alice transmit $x^n$ through the c-q wiretap channel $\mathcal{N}^{\otimes n}$.
    \item Upon receiving the channel output state, Bob performs a positive operator-valued measure $\Pi := \{\Pi^{m,l}_B\}_{m,l\in[ML]}$ to obtain an outcome $\hat{m},\hat{l}$ so as to get the message $\hat{m}$.
\end{enumerate}
Let $\log{M} = nR$, $\log{L} = nR_1$, $\log{K} = nR_2$.
The error probability of Bob's decoding and the Eve's distinguishability for different messages are defined as  
\begin{align}
        &\varepsilon(f,k) := \frac{1}{M}\sum_{m\in [M]}\Pr[\hat{m}\neq m];\\
		&d_{\text{WT}}(f,k) \\
  & := \frac{1}{2}\bigg\| \frac{1}{M}\sum_{m}\ket{m}\bra{m}\otimes\bigg(\frac{1}{L}\sum_{x^n:f(x^n)= (m,\cdot,k)}\sigma^{x^n}_{E^n}\bigg) \\
 &\quad - \bigg(\frac{1}{M}\sum_{m}\ket{m}\bra{m}\bigg)\otimes \bigg(\frac{1}{ML}\sum_{x^n:f(x^n)= (\cdot,\cdot,k)}\sigma^{x^n}_{E^n}\bigg) \bigg\|_1.
\end{align}
We obtain the following result.
\begin{theo}[Secrecy exponent for wiretap channel coding] \label{theo:wireach}
	Consider a classical-quantum wiretap channel $x\mapsto \sigma_{BE}^x$.
    For any probability distribution $p$ and $\sigma_{XBE} := \sum_{x\in\mathcal{X}} p(x) |x\rangle \langle x| \otimes \sigma_{BE}^x$, let $R < I(X:B)_{\sigma}$.
    Then, there exists a sequence of constant composition codes such that $\mathds{E}_{f,k} [\varepsilon]$ asymptotically vanishes and
	\begin{align}
		&\lim_{n\to \infty}-\frac{1}{n}\log \mathds{E}_{f,k} \left[d_{\textnormal{WT}}(f,k) \right] \\
  &\geq  \sup_{\alpha\in(1,2)}\frac{\alpha-1}{\alpha}\left(I(X:B)_\sigma - \Breve{I}^*_{\alpha}(X:E)_{\sigma}- R\right).
	\end{align}
    Here, 
    the expectation is respect to choosing $f: T^n_p \to [M]\times [L] \times [K]$ 
in step $1$ and $k$ in step $2$ randomly uniformly.

    Moreover, the error exponent is positive if and only if $R < I(X:B)_\sigma - I(X:E)_\sigma$.
\end{theo}
\begin{remark}
Theorem~\ref{theo:wireach} generalizes \cite[Theorem 28]{Bei18} from the scenario of classical side information to that of quantum side information. 
Further, our result enjoys two merits.
Firstly, we sharpen the sub-exponential terms $\e^{o(n)}$ in \cite[Theorem 28]{Bei18} to the polynomial prefactor $(n+1)^{\frac12|\mathcal{X}|}$ since we directly obtain the Augustin information without resorting to a regularization property given in \cite[Lemma 27]{Bei18}.
Secondly, we allow quantum side information to be on an infinite-dimensional Hilbert space since we do not require joint typicality.
\end{remark}

\noindent The key to proving this theorem is to let $R+R_1$ be just small enough so that Bob can decode $(m,l)$ with high accuracy using constant-composition codes. 
Plus, the function $f$ can be seen as an $L$-to-$1$ regular random binning function mapping from $T^n_p$ to $[MK]$ in order to extract randomness. 

\begin{proof}[Proof of Theorem~\ref{theo:wireach}]
We adopt random constant composition codes for c-q channel coding for message rate $R+R_1$ by
\cite[Proposition~A.1]{CHDH2-2018}. 
For $R+ R_1 < I(X:B)_{\sigma}$, Bob can decode the message $(m,l)$ (and accordingly $m$) with average error probability bounded as follows:
\begin{align}
\mathds{E}_{f,k} [\varepsilon] \leq 6(n+1)^{|\mathcal{X}|}\e^{-n\sup_{\alpha\in (\sfrac12, 1)}\frac{1-\alpha}{\alpha}\left(\Breve{I}^{\uparrow}_{2-\sfrac{1}{\alpha}}(X:B)_\sigma -R - R_1 \right)}.
\end{align}
Since $f$ acts as a random binning function that maps $T^n_p$ to $[MK]$ (or $[K]$), we invoke Theorem~\ref{PAdirect} to obtain that for $\Breve{\sigma}_{X^nE^n} = \sum_{x^n\in T^n_p} \frac{1}{|T^n_p|}\ket{x^n}\bra{x^n}\otimes \sigma^{x^n}_{E^n}$,
\begin{align}
&d_{\text{PA}}(\Breve{\sigma}_{X^nE^n}, R+R_2) \leq \e^{-n \sup\limits_{\alpha\in(1,2)}\frac{\alpha-1}{\alpha}\left(\frac{\log|T^n_p|}{n} - \Breve{I}^*_{\alpha}(X:E)_{\sigma}- R-R_2\right)}\\
&d_{\text{PA}}(\Breve{\sigma}_{X^nE^n}, R_2) \leq
\e^{-n \sup\limits_{\alpha\in(1,2)}\frac{\alpha-1}{\alpha}\left(\frac{\log|T^n_p|}{n} - \Breve{I}^*_{\alpha}(X:E)_{\sigma}-R_2\right)}.
\end{align}
Moreover, 
\begin{align}
&\mathds{E}_{f,k} d_{\text{WT}}(f,k)\\
& = \mathds{E}_{f} \frac{1}{K}\sum_{k} d_{\text{WT}}(f,k)\\
 & \leq \frac{1}{2}\mathds{E}_{f}\bigg\| \frac{1}{MK}\sum_{m,k}\ket{m,k}\bra{m,k}\otimes\bigg(\frac{1}{L}\sum_{x^n:f(x^n)= (m,\cdot,k)}\sigma^{x^n}_{E^n}\bigg) \\
 &\quad - \bigg(\frac{1}{MK}\sum_{m,k}\ket{m,k}\bra{m,k}\bigg)\otimes \bigg(\frac{1}{|T^n_p|}\sum_{x^n\in T^n_p}\sigma^{x^n}_{E^n} \bigg) \bigg\|_1\\
 &\quad + \frac{1}{2}\mathds{E}_{f}\bigg\| \frac{1}{K}\sum_{k}\ket{k}\bra{k}\otimes\bigg(\frac{1}{ML}\sum_{x^n:f(x^n)= (\cdot,\cdot,k)}\sigma^{x^n}_{E^n}\bigg) \\
 &\quad - \bigg(\frac{1}{K}\sum_{k}\ket{k}\bra{k}\bigg)\otimes \bigg(\frac{1}{|T^n_p|}\sum_{x^n\in T^n_p}\sigma^{x^n}_{E^n} \bigg) \bigg\|_1\\
 & =  d_{\text{PA}}(\Breve{\sigma}_{X^nE^n}, R+R_2) + d_{\text{PA}}(\Breve{\sigma}_{X^nE^n}, R_2),
\end{align}
where in the inequality we use triangle inequality on $d_{\text{WT}}(f,k)$ and equivalent adjustment of  classical systems $m,k$ under trace norm.
By setting $R_2 = \frac1n \log|T^n_p| - I(X:B)_{\sigma} +\delta$ where $\delta>0$ is small, we have that $\lim_{n\to\infty}\mathds{E}_{f,k}[\varepsilon] =0$ and that
\begin{align}
&\lim_{n\to \infty}-\frac{1}{n}\log \mathds{E}_{f,k} \left[d_{\textnormal{WT}}(f,k) \right]\\
& \geq  \sup_{\alpha\in(1,2)}\frac{\alpha-1}{\alpha}\left(I(X:B)_\sigma - \Breve{I}^*_{\alpha}(X:E)_{\sigma}- R\right)
\end{align}
The proof is done by taking $\delta \to 0$.
\end{proof}

\section{Conclusion}\label{conclude}
We exploit the regular random binning method in quantum information-theoretic security. 
For privacy amplification against quantum side information when constant-type sources are being used, we establish the associated secrecy exponents in terms of quantum Augustin information. We also recover an error exponent for i.i.d. sources that matches well to existed results. Moreover, our result of privacy amplification applies to classical-quantum wiretap channel coding where the eavesdropper may access to quantum side information,

\newpage
\bibliographystyle{myIEEEtran}
\bibliography{reference.bib}

\begin{thebibliography}{10}
\providecommand{\url}[1]{#1}
\csname url@samestyle\endcsname
\providecommand{\newblock}{\relax}
\providecommand{\bibinfo}[2]{#2}
\providecommand{\BIBentrySTDinterwordspacing}{\spaceskip=0pt\relax}
\providecommand{\BIBentryALTinterwordstretchfactor}{4}
\providecommand{\BIBentryALTinterwordspacing}{\spaceskip=\fontdimen2\font plus
\BIBentryALTinterwordstretchfactor\fontdimen3\font minus \fontdimen4\font\relax}
\providecommand{\BIBforeignlanguage}[2]{{%
\expandafter\ifx\csname l@#1\endcsname\relax
\typeout{** WARNING: IEEEtran.bst: No hyphenation pattern has been}%
\typeout{** loaded for the language `#1'. Using the pattern for}%
\typeout{** the default language instead.}%
\else
\language=\csname l@#1\endcsname
\fi
#2}}
\providecommand{\BIBdecl}{\relax}
\BIBdecl

\bibitem{SIAM17}
C.~H. Bennett, G.~Brassard, and J.-M. Robert, ``Privacy amplification by public discussion,'' \emph{SIAM Journal on Computing, 17(2):210-229}, 1988.

\bibitem{ILL89}
R.~Impagliazzo, L.~A. Levin, and M.~Luby, ``Pseudo-random generation from one-way functions,'' in \emph{Proceedings of the twenty-first annual {ACM} symposium on Theory of computing - {STOC} {\textquotesingle}89}.\hskip 1em plus 0.5em minus 0.4em\relax {ACM} Press, 1989.

\bibitem{Ren05}
R.~Renner, ``Security of quantum key distribution,'' \href{http://dx.doi.org/10.3929/ethz-a-005115027}{\emph{Ph.D.~Thesis (ETH), arXiv:quant-ph/0512258}}, \href{http://dx.doi.org/10.3929/ethz-a-005115027}{2005}.

\bibitem{WH13}
S.~Watanabe and M.~Hayashi, ``Non-asymptotic analysis of privacy amplification via {R{\'e}nyi} entropy and inf-spectral entropy,'' in \emph{2013 {IEEE} International Symposium on Information Theory}.\hskip 1em plus 0.5em minus 0.4em\relax {IEEE}, jul 2013.

\bibitem{Hay13}
M.~Hayashi, ``Tight exponential analysis of universally composable privacy amplification and its applications,'' \href{http://dx.doi.org/10.1109/tit.2013.2278971}{\emph{{IEEE} Transactions on Information Theory}}, \href{http://dx.doi.org/10.1109/tit.2013.2278971}{vol.~59}, \href{http://dx.doi.org/10.1109/tit.2013.2278971}{no.~11}, \href{http://dx.doi.org/10.1109/tit.2013.2278971}{pp. 7728--7746}, \href{http://dx.doi.org/10.1109/tit.2013.2278971}{nov 2013}.

\bibitem{PR14}
C.~Portmann and R.~Renner, ``Cryptographic security of quantum key distribution,'' \emph{arXiv:1409.3525}, 2014.

\bibitem{HW16}
M.~Hayashi and S.~Watanabe, ``Uniform random number generation from {Markov} chains: Non-asymptotic and asymptotic analyses,'' \href{http://dx.doi.org/10.1109/tit.2016.2530084}{\emph{{IEEE} Transactions on Information Theory}}, \href{http://dx.doi.org/10.1109/tit.2016.2530084}{vol.~62}, \href{http://dx.doi.org/10.1109/tit.2016.2530084}{no.~4}, \href{http://dx.doi.org/10.1109/tit.2016.2530084}{pp. 1795--1822}, \href{http://dx.doi.org/10.1109/tit.2016.2530084}{apr 2016}.

\bibitem{Hay16}
M.~Hayashi, ``Security analysis of {$\varepsilon$}-almost dual universal{$_2$} hash functions: Smoothing of min entropy versus smoothing of {R{\'{e}}nyi} entropy of order 2,'' \href{http://dx.doi.org/10.1109/tit.2016.2535174}{\emph{{IEEE} Transactions on Information Theory}}, \href{http://dx.doi.org/10.1109/tit.2016.2535174}{vol.~62}, \href{http://dx.doi.org/10.1109/tit.2016.2535174}{no.~6}, \href{http://dx.doi.org/10.1109/tit.2016.2535174}{pp. 3451--3476}, \href{http://dx.doi.org/10.1109/tit.2016.2535174}{jun 2016}.

\bibitem{Tom16}
M.~Tomamichel, \emph{Quantum Information Processing with Finite Resources}.\hskip 1em plus 0.5em minus 0.4em\relax Springer International Publishing, 2016.

\bibitem{Dup21}
F.~Dupuis, ``Privacy amplification and decoupling without smoothing,'' \emph{arXiv:2105.05342 [quant-ph]}, 2021.

\bibitem{HT15}
M.~Hayashi and T.~Tsurumaru, ``More efficient privacy amplification with less random seeds via dual universal hash function,'' \emph{arXiv:1311.5322 [quant-ph]}, 2015.

\bibitem{Hay12_}
M.~Hayashi, ``Precise evaluation of leaked information with secure randomness extraction in the presence of quantum attacker,'' \emph{Communications in Mathematical Physics}, vol. 333, pp. 335--350, 2012.

\bibitem{TSS+11}
M.~Tomamichel, C.~Schaffner, A.~Smith, and R.~Renner, ``Leftover hashing against quantum side information,'' \href{http://dx.doi.org/10.1109/tit.2011.2158473}{\emph{{IEEE} Transactions on Information Theory}}, \href{http://dx.doi.org/10.1109/tit.2011.2158473}{vol.~57}, \href{http://dx.doi.org/10.1109/tit.2011.2158473}{no.~8}, \href{http://dx.doi.org/10.1109/tit.2011.2158473}{pp. 5524--5535}, \href{http://dx.doi.org/10.1109/tit.2011.2158473}{aug 2011}.

\bibitem{Hay11}
M.~Hayashi, ``Exponential decreasing rate of leaked information in universal random privacy amplification,'' \href{http://dx.doi.org/10.1109/tit.2011.2110950}{\emph{{IEEE} Transactions on Information Theory}}, \href{http://dx.doi.org/10.1109/tit.2011.2110950}{vol.~57}, \href{http://dx.doi.org/10.1109/tit.2011.2110950}{no.~6}, \href{http://dx.doi.org/10.1109/tit.2011.2110950}{pp. 3989--4001}, \href{http://dx.doi.org/10.1109/tit.2011.2110950}{2011}.

\bibitem{Hay2112}
J.~Wu, G.-L. Long, and M.~Hayashi, ``Quantum secure direct communication with private dense coding using general preshared quantum state,'' \emph{arXiv:2112.15113 [quant-ph]}, 2021.

\bibitem{IEEE64}
J.~M. Renes, ``On privacy amplification, lossy compression, and their duality to channel coding,'' \emph{IEEE Transactions on Information Theory 64, 7792}, 2018.

\bibitem{IEEE57}
J.~M. Renes and R.~Renner, ``Noisy channel coding via privacy amplification and information reconciliation,'' \emph{IEEE Transactions on Information Theory 57, 7377}, 2011.

\bibitem{Proc465}
A.~Abeyesinghe, I.~Devetak, P.~Hayden, and A.~Winter, ``The mother of all protocols: restructuring quantum informations family tree,'' \emph{Proc. R. Soc. A, \textbf{465}(2108), 2537-2563}, 2009.

\bibitem{Proc467}
J.~M. Renes, ``Duality of privacy amplification against quantum adversaries and data compression with quantum side information,'' \emph{Proc. Roy. Soc. A vol. 467 no. 2130, pp. 1604-1623}, 2011.

\bibitem{Tsurumaru2021EquivalenceOT}
T.~Tsurumaru, ``Equivalence of three classical algorithms with quantum side information: Privacy amplification, error correction, and data compression,'' \emph{IEEE Transactions on Information Theory, Volume 68, Issue 2, 1016 - 1031}, 2022.

\bibitem{MH333}
M.~Hayashi, ``Precise evaluation of leaked information with universal2 privacy amplification in the presence of quantum attacker,'' \emph{Communications in Mathematical Physics, Volume 333, Issue 1, pp 335-350}, 2015.

\bibitem{TH13}
M.~Tomamichel and M.~Hayashi, ``A {Hierarchy} of {Information} {Quantities} for {Finite} {Block} {Length} {Analysis} of {Quantum} {Tasks},'' \href{http://dx.doi.org/10.1109/TIT.2013.2276628}{\emph{IEEE Transactions on Information Theory}}, \href{http://dx.doi.org/10.1109/TIT.2013.2276628}{vol.~59}, \href{http://dx.doi.org/10.1109/TIT.2013.2276628}{no.~11}, \href{http://dx.doi.org/10.1109/TIT.2013.2276628}{pp. 7693--7710}, \href{http://dx.doi.org/10.1109/TIT.2013.2276628}{Nov. 2013}, 00112 arXiv: 1208.1478.

\bibitem{LYH23}
K.~Li, Y.~Yao, and M.~Hayashi, ``Tight exponential analysis for smoothing the max-relative entropy and for quantum privacy amplification,'' \emph{IEEE Transactions on Information Theory, vol. 69, no. 3, pp. 1680-1694, March}, 2023.

\bibitem{YW66}
Y.~Watanabe, ``Randomness extraction via a quantum generalization of the conditional collision entropy,'' \emph{IEEE Transactions on Information Theory, Volume 66, Issue 2}, 2020.

\bibitem{SD22}
R.~Salzmann and N.~Datta, ``Total insecurity of communication via strong converse for quantum privacy amplification,'' \emph{arXiv:2202.11090 [quant-ph]}, 2022.

\bibitem{SGC22a}
Y.-C. Shen, L.~Gao, and H.-C. Cheng, ``Strong converse for privacy amplification against quantum side information,'' \emph{arXiv:2202.10263 [quant-ph]}, 2022.

\bibitem{Yas14}
M.~H. Yassaee, M.~R. Aref, and A.~Gohari, ``Achievability proof via output statistics of random binning,'' \href{http://dx.doi.org/10.1109/tit.2014.2351812}{\emph{{IEEE} Transactions on Information Theory}}, \href{http://dx.doi.org/10.1109/tit.2014.2351812}{vol.~60}, \href{http://dx.doi.org/10.1109/tit.2014.2351812}{no.~11}, \href{http://dx.doi.org/10.1109/tit.2014.2351812}{pp. 6760--6786}, \href{http://dx.doi.org/10.1109/tit.2014.2351812}{nov 2014}.

\bibitem{Wyn75-1}
A.~Wyner, ``The commen information of two dependent random variables,'' \emph{IEEE International Symposium on Information Theory, vol. 21, no. 2, pp. 163-179, Mar}, 1975.

\bibitem{MJ13}
M.~R. Bloch and J.~N. Laneman, ``Strong secrecy from channel resolvability,'' \emph{IEEE International Symposium on Information Theory, vol. 59, no. 12, pp. 8077-8098, Dec.}, 2013.

\bibitem{TS93}
T.~Han and S.~Verdu, ``Approximation theory of output states,'' \emph{IEEE International Symposium on Information Theory, vol. 39, no. 3, pp. 752-772, May}, 1993.

\bibitem{SM14}
S.~Watanabe and M.~Hayashi, ``Strong converse and second-order asymptotics of channel resolvability,'' \emph{IEEE International Symposium on Information Theory, June}, 2014.

\bibitem{DW03}
I.~Devetak and A.~Winter, ``Classical data compression with quantum side information,'' \href{http://dx.doi.org/10.1103/physreva.68.042301}{\emph{Physical Review A}}, \href{http://dx.doi.org/10.1103/physreva.68.042301}{vol.~68}, \href{http://dx.doi.org/10.1103/physreva.68.042301}{no.~4}, \href{http://dx.doi.org/10.1103/physreva.68.042301}{Oct 2003}.

\bibitem{DW05}
------, ``Distillation of secret key and entanglement from quantum states,'' \href{http://dx.doi.org/10.1098/rspa.2004.1372}{\emph{Proceedings of the Royal Sciety A}}, \href{http://dx.doi.org/10.1098/rspa.2004.1372}{vol. 461}, \href{http://dx.doi.org/10.1098/rspa.2004.1372}{no. 2053}, \href{http://dx.doi.org/10.1098/rspa.2004.1372}{pp. 207--235}, \href{http://dx.doi.org/10.1098/rspa.2004.1372}{2005}.

\bibitem{LD09}
Z.~Luo and I.~Devetak, ``Channel simulation with quantum side information,'' \emph{{IEEE} Transactions on Information Theory}, vol.~55, no.~3, pp. 1331--1342, mar 2009.

\bibitem{Bei18}
M.~M. Mojahedian, S.~Beigi, A.~Gohari, M.~H. Yassaee, and M.~R. Aref, ``A correlation measure based on vector-valued {$L_p$}-norms,'' \href{http://dx.doi.org/10.1109/tit.2019.2937099}{\emph{{IEEE} Transactions on Information Theory}}, \href{http://dx.doi.org/10.1109/tit.2019.2937099}{vol.~65}, \href{http://dx.doi.org/10.1109/tit.2019.2937099}{no.~12}, \href{http://dx.doi.org/10.1109/tit.2019.2937099}{pp. 7985--8004}, \href{http://dx.doi.org/10.1109/tit.2019.2937099}{dec 2019}.

\bibitem{DW14}
M.~Dalai and A.~Winter, ``Constant compositions in the sphere packing bound for classical-quantum channels,'' \emph{IEEE Transactions on Information Theory}, vol.~63, no.~9, Sept 2017.

\bibitem{MO18}
M.~Mosonyi and T.~Ogawa, ``Divergence radii and the strong converse exponent of classical-quantum channel coding with constant compositions,'' \href{http://dx.doi.org/10.1109/tit.2020.3041205}{\emph{{IEEE} Transactions on Information Theory}}, \href{http://dx.doi.org/10.1109/tit.2020.3041205}{vol.~67}, \href{http://dx.doi.org/10.1109/tit.2020.3041205}{no.~3}, \href{http://dx.doi.org/10.1109/tit.2020.3041205}{pp. 1668--1698}, \href{http://dx.doi.org/10.1109/tit.2020.3041205}{mar 2021}.

\bibitem{CGH18}
H.-C. Cheng, L.~Gao, and M.-H. Hsieh, ``Properties of noncommutative r{\'{e}}nyi and {Augustin} information,'' \href{http://dx.doi.org/10.1007/s00220-022-04319-8}{\emph{Communications in Mathematical Physics}}, \href{http://dx.doi.org/10.1007/s00220-022-04319-8}{feb 2022}.

\bibitem{SGC22b}
Y.-C. Shen, H.-C. Cheng, and L.~Gao, ``Optimal second-order rates for quantum soft covering and privacy amplification,'' \emph{arXiv:2202.11590 [quant-ph]}, 2022.

\bibitem{CHDH2-2018}
H.-C. Cheng, E.~P. Hanson, N.~Datta, and M.-H. Hsieh, ``Duality between source coding with quantum side information and classical-quantum channel coding,'' \href{http://dx.doi.org/10.1109/tit.2022.3182748}{\emph{IEEE Transactions on Information Theory (Early Access)}}, \href{http://dx.doi.org/10.1109/tit.2022.3182748}{2018}.

\bibitem{Pet86}
D.~Petz, ``Quasi-entropies for finite quantum systems,'' \href{http://dx.doi.org/10.1016/0034-4877(86)90067-4}{\emph{Reports on Mathematical Physics}}, \href{http://dx.doi.org/10.1016/0034-4877(86)90067-4}{vol.~23}, \href{http://dx.doi.org/10.1016/0034-4877(86)90067-4}{no.~1}, \href{http://dx.doi.org/10.1016/0034-4877(86)90067-4}{pp. 57--65}, \href{http://dx.doi.org/10.1016/0034-4877(86)90067-4}{Feb 1986}.

\bibitem{MDS+13}
M.~M{\"u}ller-Lennert, F.~Dupuis, O.~Szehr, S.~Fehr, and M.~Tomamichel, ``On quantum {R{\'e}nyi} entropies: A new generalization and some properties,'' \href{http://dx.doi.org/10.1063/1.4838856}{\emph{Journal of Mathematical Physics}}, \href{http://dx.doi.org/10.1063/1.4838856}{vol.~54}, \href{http://dx.doi.org/10.1063/1.4838856}{no.~12}, \href{http://dx.doi.org/10.1063/1.4838856}{p. 122203}, \href{http://dx.doi.org/10.1063/1.4838856}{2013}.

\bibitem{WWY14}
M.~M. Wilde, A.~Winter, and D.~Yang, ``Strong converse for the classical capacity of entanglement-breaking and {Hadamard} channels via a sandwiched {R{\'{e}}nyi} relative entropy,'' \href{http://dx.doi.org/10.1007/s00220-014-2122-x}{\emph{Communications in Mathematical Physics}}, \href{http://dx.doi.org/10.1007/s00220-014-2122-x}{vol. 331}, \href{http://dx.doi.org/10.1007/s00220-014-2122-x}{no.~2}, \href{http://dx.doi.org/10.1007/s00220-014-2122-x}{pp. 593--622}, \href{http://dx.doi.org/10.1007/s00220-014-2122-x}{Jul 2014}.

\bibitem{Ume62}
H.~Umegaki, ``Conditional expectation in an operator algebra. {IV}. entropy and information,'' \href{http://dx.doi.org/10.2996/kmj/1138844604}{\emph{Kodai Mathematical Seminar Reports}}, \href{http://dx.doi.org/10.2996/kmj/1138844604}{vol.~14}, \href{http://dx.doi.org/10.2996/kmj/1138844604}{no.~2}, \href{http://dx.doi.org/10.2996/kmj/1138844604}{pp. 59--85}, \href{http://dx.doi.org/10.2996/kmj/1138844604}{1962}.

\bibitem{MO14}
M.~Mosonyi and T.~Ogawa, ``Quantum hypothesis testing and the operational interpretation of the quantum {R{\'{e}}nyi} relative entropies,'' \href{http://dx.doi.org/10.1007/s00220-014-2248-x}{\emph{Communications in Mathematical Physics}}, \href{http://dx.doi.org/10.1007/s00220-014-2248-x}{vol. 334}, \href{http://dx.doi.org/10.1007/s00220-014-2248-x}{no.~3}, \href{http://dx.doi.org/10.1007/s00220-014-2248-x}{pp. 1617--1648}, \href{http://dx.doi.org/10.1007/s00220-014-2248-x}{Dec 2014}.

\bibitem{MO17}
------, ``Strong converse exponent for classical-quantum channel coding,'' \href{http://dx.doi.org/10.1007/s00220-017-2928-4}{\emph{Communications in Mathematical Physics}}, \href{http://dx.doi.org/10.1007/s00220-017-2928-4}{vol. 355}, \href{http://dx.doi.org/10.1007/s00220-017-2928-4}{no.~1}, \href{http://dx.doi.org/10.1007/s00220-017-2928-4}{pp. 373--426}, \href{http://dx.doi.org/10.1007/s00220-017-2928-4}{Oct 2017}.

\bibitem{CK11}
I.~Csisz{\'a}r and J.~K{\"o}rner, \emph{Information Theory: Coding Theorems for Discrete Memoryless Systems}.\hskip 1em plus 0.5em minus 0.4em\relax Cambridge University Press ({CUP}), 2011.

\bibitem{CG22}
H.-C. Cheng and L.~Gao, ``Error exponent and strong converse for quantum soft covering,'' \emph{arXiv:2202.10995 [quant-ph]}, 2022.

\end{thebibliography}

\appendix
\subsection{Proof of Theorem~\ref{scdirect}} \label{sec:proof_scdirect}
In order to proof theorem~\ref{scdirect}, we need the following lemma for bounding how randomly the method of drawing elements without replacement can be.
Let $\Omega$ be a product space. Furthermore, let $\Omega^{M} = \{(\omega_1, \cdots, \omega_M):( \forall i\in \{1,\dots,M\}, \omega_i \in \Omega) \wedge (\forall i\neq j,  \omega_i \neq \omega_j)\}$ be an $M$-fold product space s.t. each element $\omega \in \Omega$ cannot be picked more than once. Let $L_{\infty}(\Omega, \mu)$ and $L_{\infty}(\Omega^M, \mu^M)$ be probability spaces each on $\Omega$ and $\Omega^M$. 
Define the following maps:
\begin{align}
    &\pi_i : L_{\infty}(\Omega, \mu) \to L_{\infty}(\Omega^M, \mu^M),\\
    &\pi_i(f)(\omega_1, \dots,\omega_M) = f(\omega_i),\\
    &E: L_{\infty}(\Omega, \mu) \to L_{\infty}(\Omega^M, \mu^M),\\
    &E(f)(\omega_1,\dots,\omega_M) = \int_{\Omega^M}\frac{1}{M}\sum_{i=1}^M f(\omega_i) d\mu^M,\\
    &\Theta : L_{\infty}(\Omega, \mu) \to L_{\infty}(\Omega^M, \mu^M),\\
    & \Theta(f) = \frac{1}{M}\sum_{i=1}^M \pi_i(f) - E(f).
\end{align}
$\pi_i(f)$ is, $\omega_i$, the $i$-th element $\Omega^M$, acting on function $f$. Note that 
\begin{align}
E(f) =    \int_{\Omega^M}\frac{1}{M}\sum_{i=1}^M f(\omega_i) d\mu^M =\int_{\Omega} f(\omega) d\mu(\omega).
\end{align}
That is, the expectation of $E(f)$ is also the expectation of $f$ on $\Omega$ with uniform distribution.

\begin{lemm}\label{interpolation}
For any Hilbert space $\mathcal{H}$ and $1\leq p\leq 2$, 
\begin{align}
\left\|\Theta \otimes \text{id}: L_p(\Omega, S_p(\mathcal{H})) \to L_p(\Omega^M, S_p(\mathcal{H})) \right\| \leq 2^{\frac{2}{p}-1}M^{\frac{1-p}{p}},
\end{align}
where id is acting on $S_p(\mathcal{H})$. Here, whenever we mention $\Omega$ (or $\Omega^M$), we automatically imply the probability measure being $\mu$ (or $\mu^M$).
\end{lemm}
\begin{proof}
For each $i$,
\begin{align}
&\left\|\pi_i\otimes \text{id} (f)\right\|^p_{ L_p(\Omega^M, S_p(\mathcal{H}))}   \\
& = \int_{\Omega^M} \left\|\pi_i\otimes \text{id} (f)(\omega_1,\cdots, \omega_M)\right\|_p^p d \mu^M(\omega_1,\cdots, \omega_M)\\
& = \int_{\Omega^M} \left\|f(\omega_i)\right\|_p^p d \mu^M(\omega_i)\\
& = \int_{\Omega} \left\|f(\omega)\right\|_p^p d \mu(\omega)\\
& = \left\|f\right\|^p_{L_p(\Omega, S_p(\mathcal{H}))}.
\end{align}
Furthermore, by convexity of $S_p(\mathcal{H})$ norm, 
\begin{align}
&\left\|E(f)\right\|_{L_p(\Omega^M, S_p(\mathcal{H}))} \\
&= \left\|\int_{\Omega} f(\omega) d\mu(\omega)\right\|_p\\
& \overset{(a)}{\leq} \int_{\Omega} \left\|f(\omega)\right\|_p d\mu(\omega)\\
& \overset{(b)}{\leq} \left(\int_{\Omega} \left\|f(\omega)\right\|^p_p d\mu(\omega)\right)^{\sfrac{1}{p}}\\
& = \left\|f\right\|_{L_p(\Omega, S_p(\mathcal{H}))}.
\end{align}
where in (a) we use Jensen's inequality, and in (b) we use H\"{o}lder inequality.
Therefore, for $p=1$, using triangular inequality,
\begin{align}
&\left\|\Theta \otimes \text{id}(f)\right\|_{L_1(\Omega^M, S_1)} \\
&\leq \frac{1}{M}\sum_{i=1}^M \left\|\pi_i\otimes \text{id} (f)\right\|_{L_1(\Omega^M, S_1)} + \left\|E(f)\right\|_{L_1(\Omega^M, S_1(\mathcal{H}))}\\
&\leq 2 \left\|f\right\|_{L_1(\Omega, S_1(\mathcal{H}))},
\end{align}
and thus, 
\begin{align}
\left\|\Theta \otimes \text{id}: L_1(\Omega, S_1(\mathcal{H})) \to L_1(\Omega^M, S_1(\mathcal{H})) \right\| \leq 2.\label{p1}
\end{align}

Let $\hat{f}_i = \pi_i(f)-E(f)$, thus
\begin{align}
\Theta \otimes \text{id}(f) = \frac{1}{M}\sum_{i=1}^M(\pi_i(f)-E(f)) = \frac{1}{M}\sum^M_{i=1}\hat{f}_i.
\end{align}
Using Lemma \ref{orthogonal}, for $i\neq j$, 
\begin{align}
\mathds{E}(\hat{f}_i\hat{f}_j) \leq 0.
\end{align}
Moreover, for every $i$,
\begin{align}
\left\|\hat{f}_i\right\|_{L_2(\Omega^M,S_2(\mathcal{H}))} &=\left\|\pi_i(f)- E(f)\right\|_{L_2(\Omega^M,S_2(\mathcal{H}))} \\
&= \left\|f - E(f)\right\|_{L_2(\Omega,S_2(\mathcal{H}))}\\
& \leq \left\|f\right\|_{L_2(\Omega,S_2(\mathcal{H}))}.
\end{align}
where the last inequality can be proved by direct calculation.
Thus, 
\begin{align}
\left\|\Theta\otimes\text{id}(f)\right\|^2_{L_2(\Omega^M, S_2(\mathcal{H}))} &= \mathds{E}_{\mu^M}\left\|\frac{1}{M}\sum_{i=1}^{M}\hat{f}_i\right\|_2^2\\
&= \mathds{E}_{\mu^M}\left(\frac{1}{M^2}\sum_{i,j=1}^M \hat{f}_i\hat{f}_j\right)\\
& \leq \frac{1}{M^2}\sum_{i=1}^M \mathds{E}_{\mu^M}\left(\hat{f}_i\hat{f}_i\right)\\
& \leq \frac{1}{M^2}\sum_{i=1}^M \left\|\hat{f}_i\right\|^2_{L_2(\Omega^M,S_2(\mathcal{H}))}\\
& \leq \frac{1}{M}\left\|f\right\|^2_{L_2(\Omega,S_2(\mathcal{H}))}.
\end{align}
Therefore, for $p = 2$,
\begin{align}
\left\|\Theta \otimes \text{id}: L_2(\Omega, S_2(\mathcal{H})) \to L_2(\Omega^M, S_2(\mathcal{H})) \right\| \leq \frac{1}{\sqrt{M}}. \label{p2}
\end{align}
Having \eqref{p1} and \eqref{p2}, we apply the interpolation for $1\leq p \leq 2$ with $\theta = \frac{2(p-1)}{p} \in [0,1]$ to complete the proof.

\end{proof}
\begin{lemm}\label{orthogonal}
Suppose $f$ is positive semidefinite on $\Omega$. Let $\hat{f}_i= \pi_i(f)-E(f)$ as defined in the proof of Lemma \ref{interpolation}. For $i\neq j$
\begin{align}
\mathds{E}_{\mu^M}(\hat{f}_i\hat{f}_j) \leq 0.
\end{align}

\end{lemm}
\begin{proof}
By the setting of probability space $L_{\infty}(\Omega^M, \mu^m)$, $\omega_i$ and $\omega_j$ can be seen as uniformly choosing twe elements orderly from $\Omega$, drawing without replacement. 
We first consider a different scenario: Independently choose $\omega_i$ and $\omega_j$ uniformly from $\Omega$. We denote the expectation in this scenario $\mathds{E}_{\mu'}(\cdot)$.
In fact, 
\begin{align}
&\mathds{E}_{\mu'}(\hat{f}_i\hat{f}_j)\\
&= \mathds{E}_{\mu'}((f(\omega_i)-E(f))(f(\omega_j)-E(f)))\\
& = \int_{\omega_i \neq \omega_j}|(f(\omega_i)-E(f))(f(\omega_j)-E(f))| d\mu_i d\mu_j \\
&\quad + \int_{\omega_i = \omega_j}|(f(\omega_i)-E(f))(f(\omega_j)-E(f))|d\mu_i d\mu_j\\
& = c\mathds{E}_{\mu^M}(\hat{f}_i\hat{f}_j) \\
&\quad+\int_{\omega_i = \omega_j}|(f(\omega_i)-E(f))(f(\omega_j)-E(f))|d\mu_i d\mu_j,
\end{align}
for some $0<c\leq1$. Note that the second term of the last equation is $0$ if $\Omega$ is continuously distributed and no smaller than $0$ if $\Omega$ isn't.
Furthermore, 
\begin{align}
\mathds{E}_{\mu'}(\hat{f}_i\hat{f}_j) = \mathds{E}_{\mu}(\hat{f}_i)\mathds{E}_{\mu}(\hat{f}_j) = 0  
\end{align}
Using above two equations and the fact that $f$ positive semi-definite, we finish the proof. 
\end{proof}

\begin{proof}[Proof of Theorem~\ref{scdirect}]
Introduce 
\begin{align}
\Breve{\rho}_{\Omega B^n} = \sum_{x^n\in T^n_p} \mathds{1}_{A_{x^n}}\otimes \rho_{x^n} \in L_{\infty}(\Omega, B(\mathcal{H})^{\otimes n})
\end{align}
Denote $\mathcal{H}\equiv \mathcal{H}_B$. For $x^n = x_1\cdots x_n \in \mathcal{X}^n$, $\rho_{x^n} := \rho_{x_1}\otimes \cdots \otimes \rho_{x_n} \in \mathcal{B}(\mathcal{H})^{\otimes n}$. For all $x\in \mathcal{X}$, let $n_x := np_X(x)$. By definition, the Augustine information $\Breve{I}^*_{\alpha}(X:B)_{\rho}$ can be written as 
\begin{align}
&\Breve{I}^*_{\alpha}(X:B)_{\rho} \\
&= \inf_{\sigma\in \mathcal{S(\mathcal{H})}} \frac{\alpha}{\alpha-1}\sum_{x\in\mathcal{X}}\frac{n_x}{n} \log \left\|\sigma^{-\frac{1}{2\alpha'}}\rho_x \sigma^{-\frac{1}{2\alpha'}}\right\|_{\alpha}\\
& = \inf_{\sigma\in \mathcal{S(\mathcal{H})}} \frac{1}{n}\cdot\frac{\alpha}{\alpha-1} \log \left\|(\sigma^{\otimes n})^{-\frac{1}{2\alpha'}}\rho_{x^n}(\sigma^{\otimes n})^{-\frac{1}{2\alpha'}}\right\|_{\alpha}
\end{align}
for any $x_n \in T^n_p$. Therefore, we have
\begin{align}
&\Breve{I}^*_{\alpha}(X:B)_{\rho} \\
&= \inf_{\sigma\in \mathcal{S(\mathcal{H})}}\frac{1}{n}\cdot\frac{\alpha}{\alpha-1}\left\|\sigma_{\Omega B_n}^{-\frac{1}{2\alpha'}}\Breve{\rho}_{\Omega B^n}\sigma_{\Omega B_n}^{-\frac{1}{2\alpha'}}\right\|_{L_{\alpha}(\Omega, S_{\alpha}(\mathcal{H}^{\otimes n}))},
\end{align}
where $\sigma_{\Omega B^n} = \mathds{1}_{\Omega} \otimes \sigma^{\otimes n}_B$. In other words, 
\begin{align}
\e^{n \frac{\alpha-1}{\alpha}\Breve{I}^*_{\alpha}(X:B)_{\rho}} = \inf_{\sigma\in\mathcal{S}(\mathcal{H})}\left\|\sigma_{\Omega B_n}^{-\frac{1}{2\alpha'}}\Breve{\rho}_{\Omega B^n}\sigma_{\Omega B_n}^{-\frac{1}{2\alpha'}}\right\|_{L_{\alpha}(\Omega, S_{\alpha}(\mathcal{H}^{\otimes n}))}.\label{IBreve}
\end{align}
Let $M = |\Breve{\mathcal{C}}^n|$ and using the construction of Lemma~\ref{interpolation}, 
\begin{align}
\mathds{E}_{\Breve{\mathcal{C}}^n}\left\|\rho^{\Breve{\mathcal{C}}^n}_{B^n}-\Breve{\rho}_{B^n}\right\|_1 = \left\|\Theta(\Breve{\rho}_{\Omega B^n})\right\|_{L_1(\Omega^M, S_1(\mathcal{H}^{\otimes n}))}
\end{align}
Then, for any $\sigma \in \mathcal{S}(\mathcal{H})$,
\begin{align}
&\left\|\Theta(\Breve{\rho}_{\Omega B^n})\right\|_{L_1(\Omega^M, S_1(\mathcal{H}^{\otimes n}))}\\
&\leq \left\|\sigma_{ B^n}^{-\frac{1}{2\alpha'}}\Theta(\Breve{\rho}_{\Omega B^n})\sigma_{ B^n}^{-\frac{1}{2\alpha'}}\right\|_{L_{\alpha}(\Omega^M, S_{\alpha}(\mathcal{H}^{\otimes n}))}\\
& = \left\|\Theta\left(\sigma_{\Omega B^n}^{-\frac{1}{2\alpha'}}\Breve{\rho}_{\Omega B^n}\sigma_{\Omega B^n}^{-\frac{1}{2\alpha'}}\right)\right\|_{L_{\alpha}(\Omega^M, S_{\alpha}(\mathcal{H}^{\otimes n}))}\\
& \leq \left\|\Theta: L_{\alpha}(\Omega, S_{\alpha}(\mathcal{H}^{\otimes n})) \to L_{\alpha}(\Omega^M, S_{\alpha}(\mathcal{H}^{\otimes n})) \right\|\\
&\quad\cdot \left\|\sigma_{\Omega B^n}^{-\frac{1}{2\alpha'}}\Breve{\rho}_{\Omega B^n}\sigma_{\Omega B^n}^{-\frac{1}{2\alpha'}}\right\|_{L_{\alpha}(\Omega, S_{\alpha}(\mathcal{H}^{\otimes n}))}
\end{align}
where in the first inequality we use H\"older inequality $\left\|AXB\right\|_1 \leq \|A\|_{2\alpha'}\|X\|_{\alpha}\|B\|_{2\alpha'}$ with
\begin{align}
&\left\|\sigma_{ B^n}^{-\frac{1}{2\alpha'}}\right\|_{L_{2\alpha'}(\Omega^M, S_{2\alpha'}(\mathcal{H}^{\otimes n}))} \\
&= \left(\int_{\Omega^M}\left\|\sigma_{ B^n}^{-\frac{1}{2\alpha'}}\right\|_{2\alpha'}^{2\alpha'}d\mu^M\right)^{\frac{1}{2\alpha'}} = 1.
\end{align}
Using Lemma~\ref{interpolation} and equation~\eqref{IBreve}, we obtain the following result by taking infimum over $\rho \in \mathcal{S}(\mathcal{H})$:
\begin{align}
&\mathds{E}_{\Breve{\mathcal{C}}^n}\left\|\rho^{\Breve{\mathcal{C}}^n}_{B^n}- \Breve{\rho}_{B^n}\right\|_1 \\
&\leq 2^{\frac{2}{\alpha}-1}M^{\frac{1-\alpha}{\alpha}} \inf_{\sigma\in\mathcal{S}(\mathcal{H})}\left\|\sigma_{\Omega B^n}^{-\frac{1}{2\alpha'}}\Breve{\rho}_{\Omega B^n}\sigma_{\Omega B^n}^{-\frac{1}{2\alpha'}}\right\|_{L_{\alpha}(\Omega, S_{\alpha}(\mathcal{H}^{\otimes n}))}\\
& = 2^{\frac{2}{\alpha}-1}M^{\frac{1-\alpha}{\alpha}} \e^{n \cdot \frac{\alpha-1}{\alpha}\Breve{I}_{\alpha}^*(X:B)_{\rho}}\\
& \leq 2\e^{-n\frac{1-\alpha}{\alpha}(\Breve{I}_{\alpha}^*(X:B)_{\rho}-R)},
\end{align}
where $R = \frac{1}{n}\log|\mathcal{C}^n|$ and $\alpha\in (1,2)$, 
and the proof is done.
\end{proof}

\subsection{Proof of Theorem~\ref{scconverse}} \label{sec:proof_scconverse}
Let $M = |\Breve{\mathcal{C}}^n|$ and $R = \frac{1}{n}\log M$. The process of the proof follows the proof of Proposition~8 in \cite{CG22}, with part of derivation being different. 
\begin{align}
&d_{\textnormal{PA}}(\Breve{\rho}_{X^nB^n}, R)\\
&\overset{(a)}{\geq} \Tr[(\rho^{\Breve{\mathcal{C}}^n}_{B^n}-\Breve{\rho}_{B^n})(\rho^{\Breve{\mathcal{C}}^n}_{B^n}+\Breve{\rho}_{B^n})^{-\sfrac{1}{2}}\rho^{\Breve{\mathcal{C}}^n}_{B^n}(\rho^{\Breve{\mathcal{C}}^n}_{B^n}+\Breve{\rho}_{B^n})^{-\sfrac{1}{2}}]\\
&\overset{(b)}{\geq}1-2\mathds{E}_{\Breve{\mathcal{C}}^n} \frac{1}{M}\sum_{x^n\in \Breve{\mathcal{C}}^n}\Tr\big[({\rho}^{x^n}_{B^n})^{1-s}\big(\sum_{\substack{\Bar{x}^n\in\Breve{\mathcal{C}}^n\\ \Bar{x}^n\neq x^n}}{\rho}^{\Bar{x}^n}_{B^n}+M\Breve{\rho}_{B^n}\big)^s\big]\\
& \overset{(c)}{\geq} 1- 4M^s \sum_{x^n\in T^n_p}\frac{1}{|T^n_p|}\Tr\left[({\rho}^{x^n}_{B^n})^{1-s}\left(\Breve{\rho}_{B^n}\right)^s\right].\label{eq:scconve}
\end{align}
where $M = |\Breve{\mathcal{C}}^n|$ and $0<s<1$.
Here (a) and (b) follow similarly from the proof of \cite[Theorem~5]{CG22}, and (c) follows from that
\begin{align}
 &\mathds{E}_{\Breve{\mathcal{C}}^n} \frac{1}{M}\sum_{x^n\in \Breve{\mathcal{C}}^n}\Tr\big[({\rho}^{x^n}_{B^n})^{1-s}\big(\sum_{\substack{\Bar{x}^n\in\Breve{\mathcal{C}}^n\\ \Bar{x}^n\neq x^n}}{\rho}^{\Bar{x}^n}_{B^n}+M\Breve{\rho}_{B^n}\big)^s\big]\\
 & = \frac{1}{M}\sum_{x^n\in T^n_p}\mathds{E}_{\Breve{\mathcal{C}}^n|x^n}\mathds{E}_{x^n}\mathds{1}_{x^n\in\Breve{\mathcal{C}}^n}\\ &\quad\cdot\Tr\big[({\rho}^{x^n}_{B^n})^{1-s}\big(\sum_{\substack{\Bar{x}^n\in\Breve{\mathcal{C}}^n\\ \Bar{x}^n\neq x^n}}{\rho}^{\Bar{x}^n}_{B^n}+M\Breve{\rho}_{B^n}\big)^s\big]\\
  & = \frac{1}{M}\sum_{x^n\in T^n_p}\mathds{E}_{\Breve{\mathcal{C}}^n|(x^n\in \Breve{\mathcal{C}}^n)}\frac{M}{|T^n_p|}\\
  &\quad\cdot\Tr\big[({\rho}^{x^n}_{B^n})^{1-s}\big(\sum_{\substack{\Bar{x}^n\in\Breve{\mathcal{C}}^n\\ \Bar{x}^n\neq x^n}}{\rho}^{\Bar{x}^n}_{B^n}+M\Breve{\rho}_{B^n}\big)^s\big]\\
 & \overset{(a)}{\leq}  \frac{1}{M}\sum_{x^n\in T^n_p}\frac{M}{|T^n_p|}\\
 &\quad\cdot\Tr\big[({\rho}^{x^n}_{B^n})^{1-s}\big(\mathds{E}_{\Breve{\mathcal{C}}^n|(x^n\in \Breve{\mathcal{C}}^n)}\sum_{\substack{\Bar{x}^n\in\Breve{\mathcal{C}}^n\\ \Bar{x}^n\neq x^n}}{\rho}^{\Bar{x}^n}_{B^n}+M\Breve{\rho}_{B^n}\big)^s\big]\\
 & =  \frac{1}{M}\sum_{x^n\in T^n_p}\frac{M}{|T^n_p|}\\
 &\quad\cdot\Tr\bigg[({\rho}^{x^n}_{B^n})^{1-s}\bigg(\frac{M-1}{|T^n_p|-1}\sum_{\substack{\Bar{x}^n\in T^n_p\\ \Bar{x}^n\neq x^n}}{\rho}^{\Bar{x}^n}_{B^n}+M\Breve{\rho}_{B^n}\bigg)^s\bigg]\\
 & \overset{(b)}{\leq}   \frac{1}{M}\sum_{x^n\in T^n_p}\frac{M}{|T^n_p|}\Tr\left[({\rho}^{x^n}_{B^n})^{1-s}\left(2M\Breve{\rho}_{B^n}\right)^s\right]\\
& = (2M)^s \sum_{x^n\in T^n_p}\frac{1}{|T^n_p|}\Tr\left[({\rho}^{x^n}_{B^n})^{1-s}\left(\Breve{\rho}_{B^n}\right)^s\right],
\end{align}
where $\mathds{E}_{\Breve{\mathcal{C}}^n|x^n}$ (or $\mathds{E}_{\Breve{\mathcal{C}}^n|x^n\in \Breve{\mathcal{C}}^n}$) denotes the expectation of codebook when $x^n$ is determined to be in the codebook or not (or when $x^n \in \Breve{\mathcal{C}}^n$), and $\mathds{E}_{x^n}$ is the expectation of $x^n$ being in the codebook. Inequality (a) follows from Jensen's inequality for a concave function. In (b) we use the fact that $\frac{M}{|T^n_p|} > \frac{M-1}{|T^n_p|-1}$ and operator monotonicity of $(\cdot)^s$.

Note that 
\begin{align}
\Breve{\rho}_{B^n} &= \sum_{x^n\in \mathcal{X}^n}\frac{\mathds{1}_{x^n \in T^n_p}p_X^{\otimes n}(x^n)\rho^{x^n}_{B^n}}{p_X^{\otimes n}(T^n_p)}\\
& \leq \sum_{x^n\in \mathcal{X}^n}\frac{p_X^{\otimes n}(x^n)\rho^{x^n}_{B^n}}{p_X^{\otimes n}(T^n_p)}\\
& = \frac{\rho^{\otimes n}_{B}}{p_X^{\otimes n}(T^n_p)}. \label{eq:sctype}
\end{align}
By \eqref{eq:scconve} and \eqref{eq:sctype}, we have
\begin{align}
&d_{\textnormal{PA}}(\Breve{\rho}_{X^nB^n}, R) \\
&\geq 1- 4M^s \sum_{x^n\in T^n_p}\frac{1}{|T^n_p|}\Tr\left[({\rho}^{x^n}_{B^n})^{1-s}\left({\rho}^{\otimes n}_{B}\right)^s\right] (p_X^{\otimes n}(T^n_p))^{-s}\\
& = 1- 4M^{\frac{1-\alpha}{\alpha}} \e^{n\frac{\alpha-1}{\alpha}\Breve{I}^{\uparrow}_{2-\sfrac{1}{\alpha}}(X:B)_{\rho}}(p_X^{\otimes n}(T^n_p))^{\frac{\alpha-1}{\alpha}}\\
& \geq  1- 4(n+1)^{|\mathcal{X}|}\e^{-n\sup\limits_{\alpha\in (\sfrac12, 1)}\frac{1-\alpha}{\alpha}\left(\Breve{I}^{\uparrow}_{2-\sfrac{1}{\alpha}}(X:B)_\rho -R\right)}
\end{align}
where in the last inequality we use that the probability $p_X^{\otimes n}(T^n_p)$ is bounded by \cite{CK11}
\begin{align}
p_X^{\otimes n}(T^n_p) \geq (n+1)^{-|\mathcal{X}|}.
\end{align}
And we finish our proof of strong converse bound.
\qed

\subsection{Proof of Proposition~\ref{prop:bound}} \label{app:bound}

Fix any $\sigma_B \in \mathcal{S}(\mathcal{H}_B)$, we calculate
\begin{align}
    &H(p) - \sum_x p(x) D_\alpha^*(\rho_B^x \Vert \sigma_B ) 
    \notag
    \\
    &= - \sum_x p(x) \frac{1}{\alpha-1} \log \left[ p(x)^{\alpha-1} \e^{(\alpha-1) D_\alpha^*(\rho_B^x \Vert \sigma_B )} \right]
    \\
    &\geq -  \frac{1}{\alpha-1} \log \left[  \sum_x p(x)^{\alpha} \e^{(\alpha-1) D_\alpha^*(\rho_B^x  \Vert \sigma_B )} \right]
    \\
    &= - D_\alpha^*\left( \rho_{XB} \Vert \mathds{1}_X \otimes \sigma_B\right),
\end{align}
where the inequality follows from Jensen's inequality and the concavity of logarithm.

By minimizing over $\sigma_B \in \mathcal{S}(\mathcal{H}_B)$ concludes the proof of Proposition~\ref{prop:bound}. 
\qed


\end{document}